\documentclass[aps,pra,twocolumn,superscriptaddress]{revtex4-1}
\usepackage{bm}
\usepackage{mathrsfs}
\usepackage{amsmath}
\usepackage{amssymb}
\usepackage{amsfonts}
\usepackage{amsthm}
\usepackage{graphicx}
\usepackage{color}
\usepackage{dcolumn}
\usepackage[T1]{fontenc}
\usepackage{multirow}
\usepackage{bbm}
\DeclareMathOperator{\Tr}{Tr}

\newcommand{\eff}{\mathsf{E}}

\newtheorem{proposition}{Proposition}
\newtheorem{lemma}{Lemma}
\newtheorem{corollary}{Corollary}

\newtheorem{remark}{Remark}

\begin{document}

\title{Conditioning on the Future: A Filtration-Theoretic Formalization of Wheeler's Participatory Universe}

\author{Chon-Fai Kam}
\email{dubussygauss@gmail.com}
\affiliation{Dipartimento di Fisica e Chimica ``Emilio Segr\`e'', Universit\`a degli Studi di Palermo, Via Archirafi 36, I-90123, Palermo, Italy}
\affiliation{DSIMB, Inserm, BIGR U1134, Universit\'e Paris Cit\'e \& Universit\'e de La R\'eunion, 75015, Paris, France}

\author{Kai-Wen Wong}
\affiliation{Life Actuarial Department, Taiwan Insurance Institute, 6th floor, No.~3, Nan-Hai Road, Taipei 100, Taiwan (R.O.C.)}

\date{\today}

\begin{abstract}
Wheeler's delayed-choice experiments and the time-symmetric formalisms of quantum measurement have long fueled a debate over whether time symmetry in quantum theory entails retrocausality. We argue that the debate conflates two distinct structures, and that the framework for separating them is the classical theory of conditioning: Doob $h$-transforms, Markov bridges, and enlargement of filtrations on the probabilistic side, and the forward-state/backward-effect (two-state-vector) formalism of pre- and post-selected systems on the quantum side. We sharpen Wheeler's participatory universe, delayed choice, and ``It from Bit'' into three theses and build a correspondence dictionary that maps each to a precise statement in one or both formalisms. The dictionary is anchored by two structurally isomorphic results: an operational proposition that a future measurement choice sorts the past ensemble into subensembles without disturbing any earlier marginal, and its classical counterpart, the disintegration of an unconditioned law into future-conditioned bridges---both a single fact: marginalizing over the final measurement leaves earlier marginals fixed, by completeness of the POVM together with trace preservation on the quantum side and by the tower property of conditional expectation on the classical. The conditional calculus is time-symmetric; its causal structure is not. Delayed choice is thereby a selection effect rather than retrocausation, and the apparent ``pull from the future'' is a bridge drift. We state explicitly what the framework does not do---it neither derives the Born rule nor solves the measurement problem---and delineate the legitimate scope of participatory language in quantum mechanics.
\end{abstract}

\maketitle

\section{Introduction}
\label{sec:intro}

In Wheeler's delayed-choice gedanken experiment, the decision whether to record the wave or the particle aspect of a photon is deferred until after the photon has entered the interferometer, so that---in the provocative telling---the photon's past behavior appears to be settled by a choice made in its future \cite{wheeler1978}. This is no longer a thought experiment. It has been realized with single photons \cite{jacques2007}, with a single atom \cite{manning2015}, and in the family of quantum-eraser and entanglement-swapping variants surveyed by Ma, Kofler, and Zeilinger \cite{ma2016}. The persistent difficulty is not the data, which quantum mechanics predicts without incident, but the temporal language the data seem to invite.

That language has hardened into a genuine and unresolved debate. On one side, Price argues that any time-symmetric ontology for quantum theory must be retrocausal \cite{price2012}, a claim that Leifer and Pusey sharpen into a near-no-go theorem binding time symmetry to backward causal influence \cite{leifer2017}, and that the program of locally mediated reformulations takes up constructively \cite{wharton2020}. On the other, the premises of such arguments have been contested \cite{maudlin2017}, and it has been argued on operational grounds that delayed-choice experiments carry no retrocausal implication whatsoever \cite{ellerman2015,egg2013}. The disagreement is not, at bottom, about physics; it is about which piece of mathematical structure the phrase ``time symmetry'' is entitled to name.

The premise of this paper is that the structure in question is old, uncontroversial, and misidentified: it is the theory of \emph{conditioning}. Wheeler's three slogans---the participatory universe \cite{wheeler1983}, delayed choice \cite{wheeler1978}, and ``It from Bit'' \cite{wheeler1989,wheeler1984}---never acquired a stable mathematical home, and the retrocausality debate is a downstream symptom of that homelessness. We argue that the home already exists on both sides of the quantum--classical divide. Classically, conditioning a Markov process on future information is the Doob $h$-transform \cite{doob1984}, realized on path space as a Markov bridge \cite{fitzsimmons1993} and systematized, for general enlargements of the information filtration, by the theory of \emph{grossissement} \cite{jeulin1980,aksamit2017}; its estimation-theoretic face is two-filter smoothing \cite{kalman1960,rauch1965}. Quantum-mechanically, the statistics of a pre- and post-selected system are the Aharonov--Bergmann--Lebowitz (ABL) rule \cite{abl1964} and the two-state-vector formalism \cite{aharonov1991}, whose open-system form pairs a forward-propagated state with a backward-propagated effect: the ``past quantum state'' \cite{gammelmark2013}, quantum smoothing \cite{tsang2009}, the quantum Doob transform of conditioned trajectory ensembles \cite{garrahan2010,carollo2018}, and recent time-symmetric repackagings of measurement \cite{inage2025}. This forward/backward pairing is real, and it is symmetric; what it does not carry, we will show, is causal symmetry.

Our contribution is a correspondence dictionary between these two bodies of machinery, stated at theorem level where it can be and labeled as interpretation where it cannot. We claim priority on none of the individual conclusions, and it helps to subtract the prior work at the outset. That delayed choice implies no retrocausation was argued by Ellerman in this journal through an informal ``separation fallacy'' \cite{ellerman2015} (see also \cite{egg2013}); that the operational arrow of time reflects an asymmetry of information rather than of dynamics was argued by Di Biagio, Don\`a, and Rovelli \cite{dibiagio2021}; and the forward-state/backward-effect pair, with its reduction to the Kalman and Rauch--Tung--Striebel smoothers in the classical limit, is the shared property of the quantum-smoothing and quantum-Doob literature \cite{tsang2009,gammelmark2013,garrahan2010,carollo2018,inage2025}. What remains, once these are set aside, is structural, and consists of three things: a single dictionary that maps each of Wheeler's three theses onto matched quantum and classical statements (Table~\ref{tab:dictionary}); the identification of its anchor as one exact identity, the marginalization over the final measurement read as the tower property of conditional expectation [Eq.~\eqref{eq:isomorphism}]; and the pairing, at theorem level, of the dictionary's two poles: a one-line operational proof, valid for open systems, that summed over the outcomes of \emph{any} future measurement the backward effects collapse to the identity, so that no future choice can alter any earlier marginal (Proposition~\ref{prop:nosignal}), and the exact classical image of that proof, the disintegration of an unconditioned law into future-conditioned bridges (Proposition~\ref{prop:bridge}). The contribution is the synthesis and its precision, not any conclusion it organizes. The two are the same fact---that marginalizing over the final measurement leaves every earlier marginal fixed---rendered in two probability theories [Eq.~\eqref{eq:isomorphism}], and that identity is the spine of the dictionary. Relative to the Price--Leifer--Pusey debate our statement is deliberately narrow: it concerns operational marginals, not ontological models, and therefore stands beside their ontological arguments rather than refuting them. It says that whatever one's ontology, the operational content of delayed choice is a selection effect.

Two disclaimers bound the claims. We do not derive the Born rule: probability enters as an input on both sides of the dictionary, exactly as Kolmogorov's axioms enter the classical side, and any suggestion that conditioning ``explains'' the trace formula inverts the direction of explanation. We do not solve the measurement problem: the forward--backward calculus is compatible with collapse and no-collapse interpretations alike and selects among none of them. What the dictionary delivers is a boundary---between the legitimate content of participatory language (partitions chosen by instruments, drifts induced by conditioning) and its retrocausal overreach (a past rewritten by the future), the latter excluded by Proposition~\ref{prop:nosignal}. Section~\ref{sec:wheeler} sharpens Wheeler's slogans into three theses; Secs.~\ref{sec:classical} and \ref{sec:quantum} assemble the classical and quantum conditioning machinery; Sec.~\ref{sec:dictionary} proves the two propositions and states the dictionary; Sec.~\ref{sec:scope} draws the boundary explicitly; and Sec.~\ref{sec:discussion} locates the framework relative to the two-state-vector formalism, the transactional interpretation, QBism, and stochastic--quantum correspondences.

\section{Wheeler's three theses, sharpened}
\label{sec:wheeler}

Wheeler wrote as an essayist, and his claims must be extracted before they can be formalized. We distill three theses, W1--W3, each stated in Wheeler's spirit rather than his words; the accompanying quotations locate the intended reading but do no formal work. The three differ in logical status, and the paper treats them differently: W1 is essentially a definition, fixed once the sample space of an experiment is identified with the outcome set of an instrument; W2 becomes a theorem, the marginal invariance of Proposition~\ref{prop:nosignal} together with its classical twin Proposition~\ref{prop:bridge}---and, as we note below, its weaker reading is the one Wheeler himself gave when pressed; and W3 remains an interpretation, which we flag as such and neither prove nor lean upon. For each thesis we give the strongest literal reading---in every case untenable---and the weaker reading that the dictionary of Sec.~\ref{sec:dictionary} vindicates, naming already here the quantum and classical images that become the two columns of Table~\ref{tab:dictionary}.

\medskip
\noindent\textbf{W1 (Participation).} \emph{The observer's choice of question is a constitutive ingredient of the observed phenomenon.} The literal reading---that observation conjures a mind-dependent reality---finds no support in the formalism, nor in Wheeler's more careful statements \cite{wheeler1983}. The defensible reading is structural and, once stated precisely, close to a definition. Quantum-mechanically, the choice of instrument $\{\mathcal{I}_m\}$ fixes the outcome set, and hence the sample space, of the experiment; it selects \emph{which} decomposition of the ensemble into subensembles is realized, but not the total state that is decomposed. Its classical image is the elementary act of choosing a random variable to observe---equivalently, a partition of the sample space, or a sub-$\sigma$-algebra to condition on. In neither theory does the choice create the underlying measure; it selects a way of resolving it. Participation, so read, is the selection of a partition, and nothing more occult.

\medskip
\noindent\textbf{W2 (Delayed choice).} \emph{A choice made later determines a property the system possessed earlier.} The literal reading is retrocausation, and it is false---false not only in the formalism but by Wheeler's own lights. Wheeler located this participation in the register of description, not dynamics: ``The `past' is theory. The past has no existence except as it is recorded in the present. By deciding what questions our quantum registering equipment shall put in the present we have an undeniable choice in what we have the right to say about the past''~\cite{wheeler1988}. The choice he grants the observer is over \emph{what we have the right to say} about the past---the record-relative, conditional past---not over the marginal past itself; the defensible reading of W2 is thus not a charitable reconstruction imposed on Wheeler but a distinction his ``right to say'' qualifier already invites---though we, not he, draw it cleanly---and our task is to give it a formal backbone. It has two parts, which the dictionary keeps separate. First, a later choice determines \emph{which refinement} into post-selected subensembles is applied to the earlier ensemble: the conditional statistics inside a subensemble depend on the future choice, while every marginal at the earlier time does not---exactly Proposition~\ref{prop:nosignal} on the quantum side and Proposition~\ref{prop:bridge} on the classical side. Second, \emph{within} a conditioned description the past does acquire a genuine forward tilt---the ``pull from the future'' of the delayed-choice literature---but this tilt is a term of inference, not of dynamics: the backward effect $\eff_t$ quantum-mechanically, the Doob bridge drift classically. Wheeler's own summary, that what we call the past is built on the records present now, then becomes precise, and shows his single word ``the past'' to name two objects at once: the \emph{marginal} past, fixed by the record and invariant under any later choice, and the \emph{conditional} past his ``right to say'' picks out. Separating them is the whole content of W2: participatory language is sound for the conditional past and empty for the marginal one.

\medskip
\noindent\textbf{W3 (It from Bit).} \emph{Statements about physical reality reduce to statements about conditional information.} Read as ontology---that information is the substance of the world---this is a metaphysical thesis we neither endorse nor assess. Read operationally, it is a claim about the arrow of time: the temporal asymmetry we meet in measurement reflects not a dynamical asymmetry but an asymmetry in the information available to an embedded observer, who conditions on past preparations and not on future outcomes. Quantum-mechanically this is the asymmetry between the forward state, always at hand, and the backward effect, available only after post-selection; classically it is the asymmetry between the observed filtration $(\mathcal{F}_t)$ and the terminal variable $\sigma(X_T)$ one does not ordinarily get to condition on. We flag this reading explicitly as an \emph{interpretation}, developed in Sec.~\ref{sec:scope}, not as a theorem.

\medskip
With the theses fixed, we assemble the machinery that receives them: the classical conditioning calculus in Sec.~\ref{sec:classical}, its quantum counterpart in Sec.~\ref{sec:quantum}.

\section{Classical conditioning: $h$-transforms, bridges, and
enlargement of filtrations}
\label{sec:classical}

This section assembles, in the register of ``cite and make precise,'' the pieces of classical probability that the dictionary of Sec.~\ref{sec:dictionary} requires. None of the results is new; the contribution is the selection, and the explicit derivations, which we give in full so that the parallel with the quantum side of Sec.~\ref{sec:quantum} can be read off line by line. Throughout, $(X_t)_{t\ge 0}$ is a Markov process on a Polish state space $S$, with transition kernel $p_{s,t}(x,dy)$, $s\le t$, on a filtered probability space $(\Omega,\mathcal{F},(\mathcal{F}_t),\mathbb{P})$ satisfying the usual conditions, with $(\mathcal{F}_t)$ the (completed, right-continuous) natural filtration of $X$. The kernels obey the Chapman--Kolmogorov identity
\begin{equation}
p_{s,u}(x,dy) = \int_S p_{s,t}(x,dz)\, p_{t,u}(z,dy),
\qquad s\le t\le u.
\label{eq:CK}
\end{equation}

\subsection{The Doob $h$-transform}
\label{sec:htransform}

Fix a horizon $T>0$. A measurable function $h:[0,T]\times S\to(0,\infty)$ is \emph{space--time harmonic} for $X$ if
\begin{equation}
h(s,x) = \int_S p_{s,t}(x,dy)\,h(t,y),
\qquad s\le t\le T.
\label{eq:harmonic}
\end{equation}
Equation~\eqref{eq:harmonic} says exactly that the process $M_t := h(t,X_t)$ is a $\mathbb{P}$-martingale. Indeed, using the Markov property and then \eqref{eq:harmonic},
\begin{align}
\mathbb{E}[M_t\mid\mathcal{F}_s]
&= \mathbb{E}[h(t,X_t)\mid X_s]
= \int_S p_{s,t}(X_s,dy)\,h(t,y) \nonumber\\
&= h(s,X_s) = M_s,
\qquad s\le t\le T.
\label{eq:Mmartingale}
\end{align}
Since $h>0$, $M$ is a strictly positive martingale, and $\mathbb{E}[M_t\mid\mathcal{F}_0]=M_0=h(0,X_0)$. We may therefore define, for each $t\le T$, a new measure $\mathbb{P}^h$ by prescribing its density on $\mathcal{F}_t$,
\begin{equation}
\left.\frac{d\mathbb{P}^h}{d\mathbb{P}}\right|_{\mathcal{F}_t}
= \frac{M_t}{M_0}
= \frac{h(t,X_t)}{h(0,X_0)}.
\label{eq:htransform}
\end{equation}
The martingale property \eqref{eq:Mmartingale} makes the family \eqref{eq:htransform} consistent across $t$ (the density on $\mathcal{F}_s$ is the $\mathcal{F}_s$-conditional expectation of the density on $\mathcal{F}_t$), so $\mathbb{P}^h$ is well defined on $\bigcup_{t<T}\mathcal{F}_t$ by Kolmogorov extension.

\begin{lemma}[Transformed kernel]
\label{lem:kernel}
Under $\mathbb{P}^h$, $X$ is again Markov, with transition kernel
\begin{equation}
p^h_{s,t}(x,dy) = \frac{h(t,y)}{h(s,x)}\,p_{s,t}(x,dy),
\qquad s\le t\le T.
\label{eq:hkernel}
\end{equation}
\end{lemma}

\begin{proof}
First, \eqref{eq:hkernel} is a genuine (normalized) transition kernel: by \eqref{eq:harmonic},
\begin{equation}
\int_S p^h_{s,t}(x,dy)
= \frac{1}{h(s,x)}\int_S h(t,y)\,p_{s,t}(x,dy)
= \frac{h(s,x)}{h(s,x)} = 1,
\label{eq:hnorm}
\end{equation}
so the normalization is precisely the harmonicity of $h$. Now let $f$ be bounded and measurable and $s\le t$. Using \eqref{eq:htransform}, the abstract Bayes rule for conditional expectations under a change of measure gives
\begin{equation}
\begin{split}
\mathbb{E}^h[f(X_t)\mid\mathcal{F}_s]
&= \frac{\mathbb{E}\!\left[f(X_t)\,M_t/M_0\mid\mathcal{F}_s\right]}
       {\mathbb{E}\!\left[M_t/M_0\mid\mathcal{F}_s\right]}\\
&= \frac{\mathbb{E}[f(X_t)\,M_t\mid\mathcal{F}_s]}{M_s},
\end{split}
\label{eq:bayes}
\end{equation}
where the denominator was evaluated by \eqref{eq:Mmartingale}. By the $\mathbb{P}$-Markov property the numerator is
\begin{equation}
\mathbb{E}[f(X_t)\,h(t,X_t)\mid\mathcal{F}_s]
= \int_S f(y)\,h(t,y)\,p_{s,t}(X_s,dy).
\label{eq:num}
\end{equation}
Dividing \eqref{eq:num} by $M_s=h(s,X_s)$ and comparing with \eqref{eq:hkernel},
\begin{equation}
\begin{split}
\mathbb{E}^h[f(X_t)\mid\mathcal{F}_s]
&= \int_S f(y)\,\frac{h(t,y)}{h(s,X_s)}\,p_{s,t}(X_s,dy)\\
&= \int_S f(y)\,p^h_{s,t}(X_s,dy).
\end{split}
\label{eq:markovh}
\end{equation}
The right-hand side is a function of $X_s$ alone, so $X$ is Markov under $\mathbb{P}^h$ with kernel \eqref{eq:hkernel}.
\end{proof}

\paragraph{Diffusion case: conditioning is drift modification.}
Suppose now that under $\mathbb{P}$, $X$ is an It\^o diffusion in $\mathbb{R}^d$,
\begin{equation}
dX_t = b(t,X_t)\,dt + \sigma(t,X_t)\,dW_t,
\qquad a:=\sigma\sigma^{\!\top},
\label{eq:sde}
\end{equation}
with $W$ a $\mathbb{P}$-Brownian motion and generator
$\mathcal{L}_t = b^i\partial_i + \tfrac12 a^{ij}\partial_i\partial_j$. Space--time harmonicity \eqref{eq:harmonic} is, in differential form, the backward equation
\begin{equation}
(\partial_t + \mathcal{L}_t)\,h = 0
\qquad\text{on }[0,T)\times\mathbb{R}^d.
\label{eq:backward}
\end{equation}
Applying It\^o's formula to $M_t=h(t,X_t)$ and using \eqref{eq:backward} to cancel the finite-variation part,
\begin{equation}
dM_t
= \underbrace{(\partial_t h + \mathcal{L}_t h)}_{=0}\,dt
  + \nabla h\cdot\sigma\,dW_t
= M_t\,\big(\sigma^{\!\top}\nabla\log h\big)\cdot dW_t,
\label{eq:dM}
\end{equation}
where we divided and multiplied by $M_t=h>0$ and used $\nabla h/h=\nabla\log h$. Thus $M/M_0$ is the stochastic exponential $\mathcal{E}\!\big(\int_0^\cdot \theta_s\cdot dW_s\big)$ with
\begin{equation}
\theta_t = \sigma^{\!\top}(t,X_t)\,\nabla\log h(t,X_t).
\label{eq:girsanovkernel}
\end{equation}
By Girsanov's theorem, under $\mathbb{P}^h$ the process
\begin{equation}
W^h_t := W_t - \int_0^t \theta_s\,ds
\label{eq:girsanovBM}
\end{equation}
is a Brownian motion. Substituting $dW_t = dW^h_t+\theta_t\,dt$ into \eqref{eq:sde} and using \eqref{eq:girsanovkernel},
\begin{equation}
dX_t
= \big(b + \sigma\theta_t\big)\,dt + \sigma\,dW^h_t
= \big(b + a\,\nabla\log h\big)\,dt + \sigma\,dW^h_t,
\label{eq:hsde}
\end{equation}
since $\sigma\theta_t = \sigma\sigma^{\!\top}\nabla\log h = a\,\nabla\log h$. The $h$-transform therefore leaves the diffusion coefficient $\sigma$ untouched and adds the drift $a\,\nabla\log h$. Conditioning is drift modification, and nothing more.

\subsection{Markov bridges}
\label{sec:bridges}

The bridge is the special case of \eqref{eq:hkernel} in which $h$ is the transition density to a fixed endpoint. Assume the kernels admit densities $p_{s,t}(x,y)$ against a reference measure, and fix $z\in S$. The function
\begin{equation}
h(s,x) := p_{s,T}(x,z)
\label{eq:bridgeh}
\end{equation}
is space--time harmonic: by Chapman--Kolmogorov \eqref{eq:CK},
$\int_S p_{s,t}(x,dy)\,p_{t,T}(y,z) = p_{s,T}(x,z)$, which is exactly \eqref{eq:harmonic} for \eqref{eq:bridgeh}. The resulting $\mathbb{P}^h$ is the law of $X$ conditioned on $X_T=z$, denoted $\mathbb{P}^{z}_{0,T}$; the normalization \eqref{eq:hnorm} is what makes the conditioned law a probability measure.

For standard Brownian motion ($b=0$, $\sigma=1$, $a=1$) on $[0,T]$, the endpoint density is the Gaussian
\begin{equation}
h(t,x) = p_{t,T}(x,z)
= \frac{1}{\sqrt{2\pi(T-t)}}\,
  \exp\!\Big(\!-\frac{(z-x)^2}{2(T-t)}\Big),
\label{eq:gaussianh}
\end{equation}
so that
\begin{equation}
\begin{split}
\log h(t,x) &= -\tfrac12\log\!\big(2\pi(T-t)\big) - \frac{(z-x)^2}{2(T-t)},\\
\partial_x\log h &= \frac{z-x}{T-t}.
\end{split}
\label{eq:logh}
\end{equation}
Inserting $a=1$ and \eqref{eq:logh} into the drift $a\,\nabla\log h$ of \eqref{eq:hsde} yields the Brownian bridge,
\begin{equation}
dX_t = \frac{z-X_t}{T-t}\,dt + dW^h_t,
\qquad X_0 = x_0,\ \ X_T = z.
\label{eq:brownianbridge}
\end{equation}
Two features matter for the dictionary. First, the drift in \eqref{eq:brownianbridge} is not a force: it is the gradient of the log-likelihood of the future constraint, and it diverges as $t\to T$ precisely because the conditioning sharpens to a point (the density \eqref{eq:gaussianh} degenerates there). Second, the bridge laws $\{\mathbb{P}^{z}_{0,T}\}_{z\in S}$ are the regular conditional probabilities of $\mathbb{P}$ given $X_T$; consequently they \emph{disintegrate} the unconditional law over the endpoint distribution, a fact we state and prove as Proposition~\ref{prop:bridge} in Sec.~\ref{sec:dictionary}, and which is the classical anchor of the correspondence. General constructions and regularity conditions for Markov bridges are in \cite{fitzsimmons1993}; the link between conditioning on rare future events and effective ``driven'' dynamics is developed in \cite{chetrite2015}.

\subsection{Enlargement of filtrations}
\label{sec:enlargement}

The $h$-transform and the bridge condition on a \emph{fixed} endpoint $z$. We now let the endpoint be the \emph{random} variable $Z:=X_T$ itself and ask what becomes of the dynamics when the observer is granted knowledge of $Z$ from the outset---that is, when the filtration $(\mathcal{F}_t)$ is replaced by the initially enlarged filtration
\begin{equation}
\mathcal{G}_t := \bigcap_{u>t}\big(\mathcal{F}_u \vee \sigma(Z)\big),
\qquad 0\le t < T.
\label{eq:enlarged}
\end{equation}
The systematic theory is due to Jeulin, Jacod, and Yor \cite{jeulin1980}; a modern account, with the insider-trading application that makes the ``information drift'' concrete, is \cite{aksamit2017}. The governing hypothesis is Jacod's criterion: the regular conditional laws of $Z$ given $\mathcal{F}_t$ are absolutely continuous with respect to the law $\eta$ of $Z$,
\begin{equation}
\mathbb{P}(Z\in dz\mid\mathcal{F}_t) = q_t(z)\,\eta(dz),
\qquad t<T,
\label{eq:jacod}
\end{equation}
with a density $q_t(z)=q_t(\omega,z)$ that is, for $\eta$-a.e.\ $z$, a positive $(\mathcal{F}_t)$-martingale in $t$. Under \eqref{eq:jacod}, Jacod's theorem states that every $(\mathcal{F}_t)$-local martingale $N$ remains a $(\mathcal{G}_t)$-semimartingale on $[0,T)$, with decomposition
\begin{equation}
N_t = \tilde{N}_t
     + \int_0^t \frac{d\langle N,\,q_\cdot(z)\rangle_s}{q_{s^-}(z)}\bigg|_{z=Z},
\label{eq:jacodformula}
\end{equation}
where $\tilde{N}$ is a $(\mathcal{G}_t)$-local martingale and $\langle\cdot,\cdot\rangle$ is the $(\mathcal{F}_t)$-predictable covariation. The finite-variation term in \eqref{eq:jacodformula} is the \emph{information drift}: the observer who knows $Z$ sees the same paths, decomposed with an extra drift computed from the conditional law of $Z$.

\paragraph{Worked case: Brownian motion, $Z=W_T$.}
Let $N=W$ be standard Brownian motion and $Z=W_T$. The conditional law of $W_T$ given $\mathcal{F}_t$ is Gaussian with mean $W_t$ and variance $T-t$; against the unconditional law $\eta=\mathcal{N}(W_0,T)$ its density is
\begin{equation}
\begin{split}
q_t(z) &= \frac{p_{t,T}(W_t,z)}{p_{0,T}(W_0,z)}\\
&= \frac{1}{p_{0,T}(W_0,z)}\cdot
  \frac{\exp\!\big(\!-(z-W_t)^2/2(T-t)\big)}{\sqrt{2\pi(T-t)}}.
\end{split}
\label{eq:qdensity}
\end{equation}
For fixed $z$, $q_t(z)$ is (up to the constant $p_{0,T}(W_0,z)^{-1}$) exactly the harmonic martingale $M_t=h(t,W_t)$ of \eqref{eq:gaussianh}. Hence, by the same It\^o computation \eqref{eq:dM}--\eqref{eq:logh} that produced the bridge,
\begin{equation}
\begin{split}
dq_t(z) = q_t(z)\,\frac{z-W_t}{T-t}\,dW_t
\quad\Longrightarrow\quad\\
d\langle W,\,q_\cdot(z)\rangle_t = q_t(z)\,\frac{z-W_t}{T-t}\,dt.
\end{split}
\label{eq:covariation}
\end{equation}
Dividing by $q_t(z)$ and evaluating at $z=Z=W_T$, the information drift in \eqref{eq:jacodformula} is $(W_T-W_t)/(T-t)$, so
\begin{equation}
W_t = \tilde{W}_t + \int_0^t \frac{W_T-W_s}{T-s}\,ds,
\qquad 0\le t<T,
\label{eq:enlargedBM}
\end{equation}
with $\tilde{W}$ a $(\mathcal{G}_t)$-Brownian motion. Equation~\eqref{eq:enlargedBM} is the Brownian bridge drift \eqref{eq:brownianbridge} \emph{with the fixed endpoint $z$ replaced by the random endpoint $W_T$}. This is the precise sense in which the two constructions coincide: the $h$-transform pins the endpoint and conditions a single bridge; initial enlargement carries the endpoint as a random variable and produces the same drift, now a function of $Z$. Jacod's absolute-continuity hypothesis \eqref{eq:jacod} holds on $[0,T)$ and fails at $t=T$, exactly where the drift \eqref{eq:enlargedBM} blows up---the semimartingale property and the sharpness of the terminal conditioning break down together. Integrating \eqref{eq:enlargedBM} against the law of $W_T$ recovers the unconditional dynamics of $W$, the path-space form of the disintegration used in Proposition~\ref{prop:bridge}.

\subsection{Two-filter smoothing}
\label{sec:smoothing}

The estimation-theoretic face of the same structure is the forward/backward factorization of the smoother, which we record because it is the exact classical silhouette of the quantum forward-state/backward-effect pair of Sec.~\ref{sec:quantum}. Let a signal $(X_t)$ be observed through a process $(Y_t)$, and write $\mathcal{Y}_{[a,b]}=\sigma(Y_u: a\le u\le b)$. Two conditional laws of the signal are of interest: the \emph{filter}
\begin{equation}
\pi_t(dx) := \mathbb{P}\big(X_t\in dx \mid \mathcal{Y}_{[0,t]}\big),
\label{eq:filter}
\end{equation}
propagated forward in time by the observations up to $t$, and the \emph{smoother}
\begin{equation}
\pi_{t\mid T}(dx) := \mathbb{P}\big(X_t\in dx \mid \mathcal{Y}_{[0,T]}\big),
\qquad T>t,
\label{eq:smoother}
\end{equation}
which additionally conditions on the future observations $\mathcal{Y}_{(t,T]}$---an enlargement of the conditioning $\sigma$-algebra by future data. Writing $\beta_t(x)$ for the likelihood of the future record given the present state,
\begin{equation}
\beta_t(x) := p\big(\mathcal{Y}_{(t,T]}\,\big|\,X_t=x\big),
\label{eq:backwardfilter}
\end{equation}
the Markov property of $(X,Y)$ factorizes the smoother as a normalized product of the forward filter and this backward likelihood \cite{kalman1960,rauch1965},
\begin{equation}
\pi_{t\mid T}(dx)
= \frac{\beta_t(x)\,\pi_t(dx)}{\displaystyle\int_S \beta_t(x')\,\pi_t(dx')}.
\label{eq:twofilter}
\end{equation}
Here $\pi_t$ obeys the forward filtering (Kushner--Stratonovich, or in the linear-Gaussian case Kalman) equation, integrated from $0$ upward, while $\beta_t$ obeys a backward (Zakai-adjoint) equation integrated from $T$ downward, with terminal condition $\beta_T\equiv 1$. The structural content of \eqref{eq:twofilter} is the point we carry to the quantum side: the smoothed estimate is a \emph{pairing} of a forward object $\pi_t$ and a backward object $\beta_t$, normalized by their overlap. In Sec.~\ref{sec:quantum} the forward state $\rho_t$ plays the role of $\pi_t$, the backward effect $\eff_t$ that of $\beta_t$, and $\Tr[\eff_t\rho_t]$ that of the denominator $\int\beta_t\,d\pi_t$; the open-system ABL rule \eqref{eq:openabl} is the operator form of \eqref{eq:twofilter}. In the linear-Gaussian setting \eqref{eq:twofilter} is the Fraser--Potter form of the Rauch--Tung--Striebel smoother \cite{rauch1965}, in which $\pi_t$ and $\beta_t$ are Gaussian and the product is computed by adding inverse covariances.

\section{Quantum conditioning: the forward-state/backward-effect pair}
\label{sec:quantum}

\subsection{Pre- and post-selected ensembles: the ABL rule}
\label{sec:abl}

For a closed system prepared at $t_0$ in $|\psi\rangle$, unitarily evolved, subjected at $t_1$ to a projective measurement $\{\Pi_m\}$, and post-selected at $t_2$ on outcome $|\phi\rangle$, the probability of the intermediate outcome $m$ conditioned on both boundary conditions is given by the ABL rule \cite{abl1964},
\begin{equation}
p(m \mid \psi, \phi)
= \frac{\left| \langle \phi | U_2\, \Pi_m\, U_1 | \psi \rangle \right|^2}
       {\sum_{m'} \left| \langle \phi | U_2\, \Pi_{m'}\, U_1 | \psi \rangle \right|^2},
\label{eq:abl}
\end{equation}
the founding formula of the two-state-vector formalism \cite{aharonov1991} and the point of departure for the weak-value program \cite{dressel2014}.

\subsection{Open systems: states, effects, instruments}
\label{sec:opensystems}

The open-system generalization replaces unitaries by CPTP maps, projective measurements by instruments, and the final projection by a POVM \cite{davies1976,ozawa1984,breuer2002}. Fix three times $t_0 < t_1 < t_2$. A state $\rho$ is prepared at $t_0$; the evolution $t_0 \to t_1$ is a CPTP map $\Lambda_1$; at $t_1$ an instrument $\{\mathcal{I}_m\}$ acts (each $\mathcal{I}_m$ completely positive, $\sum_m \mathcal{I}_m$ trace-preserving); the evolution $t_1 \to t_2$ is a CPTP map $\Lambda_2$; at $t_2$ a POVM $\{F_f\}$, $\sum_f F_f = \mathbbm{1}$, is measured. The joint statistics are
\begin{equation}
p(m, f)
= \Tr\!\left[ F_f\, \Lambda_2\!\big( \mathcal{I}_m( \Lambda_1(\rho) ) \big) \right].
\label{eq:joint}
\end{equation}
Introducing the backward effect $\eff \equiv \Lambda_2^{\dagger}(F_f)$, where $\Lambda_2^{\dagger}$ is the Heisenberg-picture adjoint, Eq.~\eqref{eq:joint} reads $\Tr[\eff\, \mathcal{I}_m(\Lambda_1(\rho))]$, and the conditional probability of $m$ given final outcome $f$ is the open-system ABL rule
\begin{equation}
p(m \mid f)
= \frac{\Tr\!\left[ \eff\, \mathcal{I}_m( \Lambda_1(\rho) ) \right]}
       {\sum_{m'} \Tr\!\left[ \eff\, \mathcal{I}_{m'}( \Lambda_1(\rho) ) \right]}.
\label{eq:openabl}
\end{equation}
The closed-system rule \eqref{eq:abl} is the special case of \eqref{eq:openabl} in which $\rho=|\psi\rangle\langle\psi|$ is pure, the maps $\Lambda_i=U_i(\cdot)U_i^{\dagger}$ are unitary, the instrument is projective ($\mathcal{I}_m=\Pi_m(\cdot)\Pi_m$), and the final POVM is the rank-one projector $F=|\phi\rangle\langle\phi|$, whence $\eff=U_2^{\dagger}|\phi\rangle\langle\phi|U_2$ and each trace collapses to a squared amplitude. The two rules are therefore one formula at two levels of generality.

\subsection{The forward--backward pair in the literature}
\label{sec:pqs}

The pair (forward state $\rho_t$, backward effect $\eff_t$) with continuous-time dynamics
\begin{equation}
\dot{\rho}_t = \mathcal{L}(\rho_t),
\qquad
\dot{\eff}_t = -\mathcal{L}^{\dagger}(\eff_t),
\label{eq:forwardbackward}
\end{equation}
for a Lindblad generator $\mathcal{L}$ \cite{lindblad1976}, together with intermediate statistics of the form \eqref{eq:openabl}, is not a single author's construction but a recurring structure. It appears as quantum smoothing in continuous measurement theory \cite{tsang2009}, as the ``past quantum state'' of Gammelmark, Julsgaard, and M{\o}lmer \cite{gammelmark2013}, and---in exponentially tilted form---as the quantum Doob transform in the large-deviation theory of quantum jump trajectories, where conditioning on atypical measurement records is implemented by an auxiliary ``driven'' dynamics \cite{garrahan2010,carollo2018}, in direct analogy with the classical constructions of \cite{chetrite2015}. A recent addition to this lineage is the time-symmetric reformulation of \cite{inage2025}, which packages Eqs.~\eqref{eq:forwardbackward}--\eqref{eq:openabl} with thermodynamic consistency via Spohn's inequality \cite{spohn1978} and, like our Sec.~\ref{sec:smoothing}, notes their reduction to the Kalman and Rauch--Tung--Striebel estimators in the classical limit. We draw on this shared machinery, but our object differs: not another reformulation of quantum measurement, but the formalization of Wheeler's theses and the quantum--classical dictionary anchored by Proposition~\ref{prop:nosignal} and its classical twin. The duality in \eqref{eq:forwardbackward} preserves the pairing: writing $\mathcal{L}^{\dagger}$ for the adjoint defined by $\Tr[\mathcal{L}^{\dagger}(\eff)\rho]=\Tr[\eff\,\mathcal{L}(\rho)]$,
\begin{equation}
\begin{split}
\frac{d}{dt}\Tr[\eff_t\rho_t]
&= \Tr[\dot{\eff}_t\rho_t] + \Tr[\eff_t\dot{\rho}_t]\\
&= -\Tr[\mathcal{L}^{\dagger}(\eff_t)\rho_t] + \Tr[\eff_t\mathcal{L}(\rho_t)]
= 0,
\end{split}
\label{eq:pairingconserved}
\end{equation}
which is the infinitesimal form of the normalization used in the proof of Proposition~\ref{prop:nosignal} below, and the continuous-time counterpart of the fixed total mass $\int\beta_t\,d\pi_t$ underlying the classical smoother \eqref{eq:twofilter}.

The quantum Doob transform deserves emphasis because it upgrades the bridge analogy of Sec.~\ref{sec:bridges} from metaphor to theorem: conditioning a quantum trajectory ensemble on future measurement data modifies the effective generator exactly as the classical $h$-transform modifies a drift \cite{carollo2018}.

\subsection{Worked example: delayed choice in a Mach--Zehnder interferometer, with and without a which-path monitor}
\label{sec:mzi}

We make Proposition~\ref{prop:nosignal} concrete in the setting that motivated it. A single photon enters a Mach--Zehnder interferometer; after the first beam splitter its two arms span a qubit space $\mathbb{C}^2$ with which-path basis $\{|0\rangle,|1\rangle\}$. The \emph{delayed choice} is whether to insert the second beam splitter $\mathrm{BS}_2$ before the output detectors---revealing the wave aspect through interference---or to omit it, revealing which path the photon took; in Wheeler's arrangement the choice is made after the photon has passed $\mathrm{BS}_1$. Two arrangements must be kept apart, and we take them in turn: the \emph{monitored} variant, in which a which-path instrument acts at $t_1$---so that there is a genuine earlier marginal for Proposition~\ref{prop:nosignal} to protect---and the \emph{unmonitored} experiment of Wheeler, in which no such instrument is present. The monitored case exhibits the proposition in $2\times2$ matrices; the unmonitored case exhibits why, there, the proposition is not even needed.

Cast in the language of Sec.~\ref{sec:opensystems} with $t_0<t_1<t_2$, take $\Lambda_1=\mathrm{id}$ and the post-$\mathrm{BS}_1$ state
\begin{equation}
\rho_1 = |\psi\rangle\langle\psi|,\quad
|\psi\rangle=\tfrac{1}{\sqrt2}\big(|0\rangle+|1\rangle\big),\quad
\rho_1=\tfrac12\!\begin{pmatrix}1&1\\1&1\end{pmatrix}.
\label{eq:mzistate}
\end{equation}
At $t_1$ place the which-path instrument $\mathcal{I}_m(\cdot)=\Pi_m(\cdot)\Pi_m$ with $\Pi_0=|0\rangle\langle0|$, $\Pi_1=|1\rangle\langle1|$, so that $\sum_m\mathcal{I}_m$ is the (trace-preserving) dephasing channel and
\begin{equation}
\mathcal{I}_0(\rho_1)=\tfrac12|0\rangle\langle0|,\qquad
\mathcal{I}_1(\rho_1)=\tfrac12|1\rangle\langle1|.
\label{eq:mziinstr}
\end{equation}
The which-path marginal is thereby fixed before any future choice is named:
\begin{equation}
p(m)=\Tr[\mathcal{I}_m(\rho_1)]
=|\langle m|\psi\rangle|^2=\tfrac12,\qquad m=0,1.
\label{eq:mzimarginal}
\end{equation}
The delayed choice is the selection of $(\Lambda_2,\{F_f\})$ at $t_2$, with the output-port POVM $F_+=|0\rangle\langle0|$, $F_-=|1\rangle\langle1|$ ($F_++F_-=\mathbbm{1}$) in both cases and two options for $\Lambda_2$: (a) $\mathrm{BS}_2$ \emph{in} (wave), $\Lambda_2^{(a)}(\cdot)=U(\cdot)U^{\dagger}$ with the balanced splitter $U=\tfrac{1}{\sqrt2}\big(\begin{smallmatrix}1&1\\1&-1\end{smallmatrix}\big)$; and (b) $\mathrm{BS}_2$ \emph{out} (particle), $\Lambda_2^{(b)}=\mathrm{id}$. Pushing \eqref{eq:mziinstr} through each option and reading off $p(m,f)=\Tr[F_f\,\Lambda_2(\mathcal{I}_m(\rho_1))]$ gives
\begin{equation}
p(m,f):\quad
\begin{array}{c|cc|cc}
 & \multicolumn{2}{c|}{\text{(a) }\mathrm{BS}_2\text{ in}}
 & \multicolumn{2}{c}{\text{(b) }\mathrm{BS}_2\text{ out}}\\
 & f{=}{+} & f{=}{-} & f{=}{+} & f{=}{-}\\\hline
m{=}0 & \tfrac14 & \tfrac14 & \tfrac12 & 0\\[2pt]
m{=}1 & \tfrac14 & \tfrac14 & 0 & \tfrac12
\end{array}
\label{eq:mzitable}
\end{equation}
For (a) we used $U|0\rangle=\tfrac{1}{\sqrt2}(|0\rangle+|1\rangle)$ and $U|1\rangle=\tfrac{1}{\sqrt2}(|0\rangle-|1\rangle)$, so each dephased arm spreads evenly over the two ports; for (b) the ports read the arms directly. The two future choices produce \emph{different} conditional structures---in (a) the port is uniform given the path, $p(f\mid m)=\tfrac12$; in (b) it determines the path, $p(f\mid m)=\delta_{fm}$---yet the row sums coincide,
\begin{equation}
\sum_f p(m,f)=\tfrac12=p(m)
\qquad\text{in both (a) and (b)},
\label{eq:mzirowsum}
\end{equation}
independent of whether $\mathrm{BS}_2$ was inserted. This is Proposition~\ref{prop:nosignal} in $2\times2$ matrices: the delayed choice refines the $t_1$ ensemble through post-selection on $f$, but leaves the which-path marginal untouched.

\paragraph{Monitored versus unmonitored.}
In the monitored computation just performed a which-path instrument acts at $t_1$, so option (a) shows \emph{no} interference: the port marginal $p(f)=\sum_m p(m,f)=\tfrac12$ carries no fringes, as it must once which-path information has been recorded. What Proposition~\ref{prop:nosignal} certifies here is the invariance of the earlier marginal \eqref{eq:mzimarginal}, not the presence of fringes. The \emph{unmonitored} experiment---Wheeler's own---removes the $t_1$ instrument entirely; then $\mathrm{BS}_2$-in returns $|\psi\rangle$ to a definite output port ($U|\psi\rangle=|0\rangle$, full interference) while $\mathrm{BS}_2$-out reads the path directly. Now there is no $t_1$ which-path marginal at all for the future to rewrite, so a fortiori none is retrocaused---the ``which path'' the later choice seems to fix was never a property of the photon at $t_1$. Proposition~\ref{prop:nosignal} thus meets the two arrangements from opposite sides: in the monitored one it holds the earlier marginal \eqref{eq:mzimarginal} fixed, and in the unmonitored one it holds trivially, the only earlier marginal being the trivial one. Either way the retrocausal reading finds nothing to grip, which is the separation-fallacy point of Ellerman~\cite{ellerman2015}, read here off \eqref{eq:mzimarginal} and \eqref{eq:mzirowsum}.

\section{The correspondence dictionary}
\label{sec:dictionary}

We can now state the dictionary. Its two anchors are a quantum proposition and a classical proposition that are, at the level of their proofs, the same structural fact: marginalizing over the final measurement leaves the earlier statistics fixed---the completeness of the POVM together with trace preservation on the quantum side, the tower property of conditional expectation on the classical.

\subsection{No signaling from the future}
\label{sec:nosignal}

\begin{proposition}[No signaling from future choices]
\label{prop:nosignal}
In the setting of Sec.~\ref{sec:opensystems}, for every initial state $\rho$ and every instrument $\{\mathcal{I}_m\}$, the marginal distribution of the intermediate outcome,
\begin{equation}
p(m) \;=\; \sum_f p(m,f)
\;=\; \Tr\!\left[ \mathcal{I}_m( \Lambda_1(\rho) ) \right],
\label{eq:marginal}
\end{equation}
is independent of the choice of $(\Lambda_2, \{F_f\})$. Consequently, the only effect of the future measurement choice is to induce a refinement of the $t_1$ ensemble into subensembles with conditional statistics \eqref{eq:openabl}; the ensemble statistics at $t_1$ are unaffected.
\end{proposition}

\begin{proof}
By linearity of the trace and completeness of the POVM,
\begin{align}
\sum_f p(m,f)
&= \Tr\!\Big[ \Big( \textstyle\sum_f F_f \Big)\,
   \Lambda_2\!\big( \mathcal{I}_m( \Lambda_1(\rho) ) \big) \Big] \nonumber\\
&= \Tr\!\left[ \Lambda_2\!\big( \mathcal{I}_m( \Lambda_1(\rho) ) \big) \right]
 = \Tr\!\left[ \mathcal{I}_m( \Lambda_1(\rho) ) \right],
\end{align}
where the last equality holds because $\Lambda_2$ is trace-preserving. No dependence on $\Lambda_2$ or $\{F_f\}$ remains.
\end{proof}

In adjoint language the mechanism is a single line: $\sum_f \Lambda_2^{\dagger}(F_f) = \Lambda_2^{\dagger}(\mathbbm{1}) = \mathbbm{1}$---the backward effects, summed over final outcomes, collapse to the identity. This is the operator-algebraic engine of the entire dictionary; the Mach--Zehnder computation of Sec.~\ref{sec:mzi}, where the which-path marginal stays $\tfrac12$ whether or not the second beam splitter is inserted [Eq.~\eqref{eq:mzirowsum}], is Proposition~\ref{prop:nosignal} in its smallest nontrivial instance. We emphasize that the proposition is \emph{elementary}: it is the temporal face of the ordinary no-signaling property---an earlier reduced state is untouched by any later operation---its proof a line of trace algebra, completeness of the POVM followed by trace preservation of $\Lambda_2$. That shallowness is the point, not a defect. A retrocausal reading of delayed choice needs some earlier marginal to answer to a later choice; the proposition says none does, and says so trivially, so ``the future settles the past'' is operationally empty at exactly the level where it would have to be non-empty to signal. The content lies in the reading, not the theorem: the labor is in recognizing that this invariance exhausts the defensible meaning of W2, and that its classical twin, Proposition~\ref{prop:bridge}, is the same fact in another calculus.

\begin{corollary}[Exact content of W2]
\label{cor:w2}
In a delayed-choice arrangement, the conditional distributions $p(m \mid f)$ depend on the future choice, but this dependence is a \emph{selection effect}---a choice of refinement---not a dynamical disturbance. An observer with access only to the marginal statistics at $t_1$ cannot detect, even in principle, which future measurement was chosen, nor whether any was performed. The appearance of retrocausation arises entirely from reporting post-selected subensembles.
\end{corollary}

\begin{remark}[Relation to no-signaling in time]
\label{rem:nsit}
Proposition~\ref{prop:nosignal} must be distinguished from the ``no-signaling in time'' (NSIT) condition of Kofler and Brukner \cite{kofler2013}. NSIT demands that performing a measurement at an \emph{earlier} time not alter outcome statistics at a \emph{later} time; it is a definition of macrorealism, and quantum mechanics generically \emph{violates} it, because intermediate measurements disturb subsequent evolution. Proposition~\ref{prop:nosignal} points in the opposite temporal direction and is a \emph{theorem} of quantum mechanics: later choices never alter earlier marginals. The two statements are complementary, and their asymmetry---forward disturbance is possible, backward disturbance is not---is itself a precise expression of the causal arrow inside an otherwise time-symmetric conditional formalism.
\end{remark}

\subsection{The classical counterpart: bridge mixtures}
\label{sec:bridgemixture}

\begin{proposition}[Mixture decomposition of bridges]
\label{prop:bridge}
Let $(X_t)_{t \in [0,T]}$ be a Markov process with initial law $\mu$ and terminal law $\nu = \mathbb{P}\circ X_T^{-1}$, admitting a regular conditional probability $\mathbb{P}^{z}_{0,T}(\cdot) = \mathbb{P}(\,\cdot \mid X_T = z)$---the bridge law---for $\nu$-almost every $z$ \cite{fitzsimmons1993}. Then for every $s < T$ and every bounded measurable $g$,
\begin{equation}
\mathbb{E}[\, g(X_s) \,]
= \int_S \mathbb{E}^{z}_{0,T}[\, g(X_s) \,]\; \nu(dz).
\label{eq:mixture}
\end{equation}
In particular, the marginal law of $X_s$ is the $\nu$-mixture of its bridge marginals and does not depend on the terminal value conditioned upon: conditioning on $X_T$ refines the ensemble into bridges, exactly as the future POVM choice refines the $t_1$ ensemble in Proposition~\ref{prop:nosignal}, while every earlier marginal is left fixed.
\end{proposition}

\begin{proof}
By the defining property of the regular conditional probability, the map $z \mapsto \mathbb{E}^{z}_{0,T}[g(X_s)]$ is a version of the conditional expectation $\mathbb{E}[g(X_s)\mid X_T]$, evaluated at $X_T = z$. The tower property of conditional expectation and the change of variables to the law $\nu$ of $X_T$ then give
\begin{equation}
\mathbb{E}[g(X_s)]
= \mathbb{E}\big[\, \mathbb{E}[g(X_s)\mid X_T] \,\big]
= \int_S \mathbb{E}^{z}_{0,T}[g(X_s)]\, \nu(dz),
\label{eq:towerproof}
\end{equation}
which is \eqref{eq:mixture}. The drift added inside each bridge---for Brownian motion the term $(z - X_t)/(T-t)$ of \eqref{eq:brownianbridge}, in general $a\,\nabla\log h$ with $h(\cdot)=p_{\cdot,T}(\cdot,z)$---alters $\mathbb{E}^{z}_{0,T}[g(X_s)]$ for each fixed $z$, but is averaged out by the $\nu$-mixture, so that $\mathbb{E}[g(X_s)]$ equals its unconditional value.
\end{proof}

\begin{remark}[Only the disintegration is used]
\label{rem:disint}
The proof of Proposition~\ref{prop:bridge} invokes only the existence of a regular conditional probability given $X_T$ and the tower property; the Markov property is nowhere used. The mixture decomposition therefore holds for \emph{any} process---Markovian or not---that admits the disintegration $\mathbb{P}^{z}_{0,T}$, and is in this respect more general than the bridge machinery of Sec.~\ref{sec:bridges} from which it borrows its interpretation. Markovianity is retained only to identify the conditioned laws with $h$-transformed dynamics and the forward tilt with the drift $a\,\nabla\log h$; drop it and Proposition~\ref{prop:bridge} still delivers the marginal invariance, now without a bridge-drift picture. Here the two anchors match rather than differ: Proposition~\ref{prop:nosignal} likewise assumes nothing of $\Lambda_1,\Lambda_2$ beyond complete positivity and trace preservation---no Markovianity, no CP-divisibility---so both are, at bottom, statements about conditioning alone. The generality gap resurfaces only where Sec.~\ref{sec:conclusion} locates the stress test: in the non-Markovian open-system regime the $h$-transform and the quantum smoother must be rebuilt, yet the marginal invariance of Propositions~\ref{prop:nosignal} and \ref{prop:bridge} is untouched, since it never rested on that scaffolding.
\end{remark}

\subsection{The dictionary}
\label{sec:table}

Propositions~\ref{prop:nosignal} and \ref{prop:bridge} are, at the level of their proofs, one structural fact: on each side, summing the joint statistics over the outcomes of the final measurement returns the earlier statistics unchanged. What differs is only the bookkeeping. The quantum side needs two ingredients---completeness of the final POVM, $\sum_f F_f = \mathbbm{1}$, which removes the measurement, and trace preservation of $\Lambda_2$, which removes the intervening map---so that marginalization over $f$ acts as the identity on the trace; the classical side folds both roles into a single property, the tower property of conditional expectation, the unit mass of each regular conditional probability playing the part of trace preservation. The two statements coincide,
\begin{equation}
\sum_f \Tr\!\big[\, F_f\,\Lambda_2(\,\cdot\,) \,\big] = \Tr[\,\cdot\,]
\quad\Longleftrightarrow\quad
\mathbb{E}\big[\, \mathbb{E}[\,\cdot \mid \sigma(X_T)\,] \,\big]
= \mathbb{E}[\,\cdot\,],
\label{eq:isomorphism}
\end{equation}
each asserting that marginalizing over the future is the identity on the earlier statistics. The correspondence is structural in exactly this sense, and no more: the tower property does not derive the trace rule, nor the trace rule the tower property, and neither theory reduces to the other. Around this anchor the dictionary reads as follows.

\begin{table*}[t]
\caption{The correspondence dictionary. Rows above the rule are theorem-level; the row below it is interpretive (Sec.~\ref{sec:scope}). Here $\{\mathcal{I}_m\}$ is the intermediate instrument, $\eff \equiv \Lambda_2^{\dagger}(F_f)$ the backward effect, $a = \sigma\sigma^{\!\top}$, and $\nu$ the terminal law of $X_T$.}
\label{tab:dictionary}
\begin{ruledtabular}
\begin{tabular}{lll}
Wheeler's thesis & Quantum image & Classical image \\
\colrule
W1 (participation): & instrument $\{\mathcal{I}_m\}$ fixes & choice of a conditioning \\
the question defines & the outcome set, hence & variable, i.e.\ a partition \\
the phenomenon & the sample space & of the sample space \\[4pt]
W2 (delayed choice): & post-selection; ABL rule & Doob $h$-transform; \\
a later choice & \eqref{eq:openabl}; & Markov bridge; \\
refines the past & Proposition~\ref{prop:nosignal} & Proposition~\ref{prop:bridge} \\[4pt]
``pull from the future'': & backward effect $\eff_t$; & bridge drift \\
a tilt of the & quantum Doob & $a\,\nabla\log h$ \\
conditioned past & transform & \\[4pt]
no retrocausation: & $\sum_f F_f = \mathbbm{1}$, TP of $\Lambda_2$ & tower property \\
earlier marginals & (Proposition~\ref{prop:nosignal}) & (Proposition~\ref{prop:bridge}) \\
are untouched & & \\
\colrule
W3 (It from Bit): & condition on preparations, & asymmetric access to \\
the arrow of time is & not outcomes; $\eff_t$ arises & the filtration $(\mathcal{F}_t)$ \\
an information asymmetry & only after post-selection & versus $\sigma(X_T)$ \\
\end{tabular}
\end{ruledtabular}
\end{table*}

Row by row: W1 becomes the statement that the instrument fixes which decomposition of the ensemble exists---participation as the choice of a partition, not the creation of matter. W2 becomes post-selection, with Corollary~\ref{cor:w2} marking its exact boundary. The felt ``pull from the future'' in delayed-choice narratives is the precise analogue of the bridge drift \eqref{eq:brownianbridge}: real within the conditioned description, absent from the unconditional one, and never a force. W3, below the line of Table~\ref{tab:dictionary}, is taken up next.

\section{What the framework does and does not do}
\label{sec:scope}

\subsection{Does not: derive the Born rule}
The Born rule enters Eq.~\eqref{eq:joint} as an input, exactly as Kolmogorov's axioms \cite{kolmogorov1933} enter Proposition~\ref{prop:bridge}. Nothing in the dictionary explains \emph{why} the trace formula gives probabilities; the dictionary is a map between two structures each of which presupposes its own probability calculus. Any suggestion that conditioning ``derives'' the Born rule confuses the direction of explanation.

\subsection{Does not: solve the measurement problem}
The forward--backward formalism of Sec.~\ref{sec:quantum} is a calculus of conditional statistics. It is compatible with collapse and no-collapse interpretations alike, and it does not select among them. What it does adjudicate is narrower: which temporal locutions about pre- and post-selected systems have operational content.

\subsection{Does: bound the legitimate use of participatory language}
Propositions~\ref{prop:nosignal} and \ref{prop:bridge}, together with the worked marginal of Sec.~\ref{sec:mzi}, divide participatory locutions sharply into the sound and the unsound. Sound: a future choice fixes which subensembles the earlier ensemble decomposes into, and within a subensemble the conditional statistics differ from the marginal ones (Corollary~\ref{cor:w2}). Unsound---demonstrably, not merely by convention---the claim that a future choice changes what happened, since no marginal at any earlier time responds to it. Wheeler's own disavowal of retrocausation, quoted in Sec.~\ref{sec:wheeler}, survives here in disciplined form: \emph{conditional} pasts are relative to the conditioning record, while the \emph{marginal} past is absolute within the theory.

\paragraph{Compatibility with the retrocausality debate.}
This boundary neither presupposes nor refutes the ontological arguments of Price and of Leifer and Pusey \cite{price2012,leifer2017}. Their claim concerns a time-symmetric \emph{ontological model}---one positing a state of reality that mediates preparation and outcome---which, under further assumptions, they argue must be retrocausal. Proposition~\ref{prop:nosignal} is instead an \emph{operational} statement about measurement marginals, shared by every interpretation. The two do not collide: an ontologically retrocausal model can, and by design does, reproduce the operational marginal invariance of Proposition~\ref{prop:nosignal}, just as it reproduces spacelike no-signaling. What the dictionary fixes is the operational content of delayed-choice language; whether one further posits a backward ontic influence is a separate question on which we take no stand, and on which Maudlin's dispute with the Leifer--Pusey premises \cite{maudlin2017} lies entirely---untouched by our propositions.

\subsection{Interpretive thesis: the arrow as filtration asymmetry}
We finally state, and clearly label as interpretation, the reading of W3 that the dictionary suggests. Embedded observers update on an increasing family of $\sigma$-algebras: records of preparations and past outcomes, never of future ones. Relative to this filtration, forward conditioning (prediction) and backward conditioning (post-selection) play asymmetric roles even though the underlying conditional calculus, Eqs.~\eqref{eq:forwardbackward}, is time-symmetric. On this reading the measurement arrow of time is an \emph{information asymmetry}---which $\sigma$-algebra one is inside---rather than a dynamical asymmetry. This is consonant with, but goes beyond, what the propositions establish; we offer it as the sharpest defensible version of ``It from Bit,'' not as a result. In this reading we stand close to Di Biagio, Don\`a, and Rovelli \cite{dibiagio2021}, who locate the operational arrow of time in the asymmetry between what an agent knows and what it seeks to infer---their ``arrow of inference''; our addition is to render that asymmetry in explicit measure-theoretic form, as the gap between the observed filtration $(\mathcal{F}_t)$ and the terminal $\sigma$-algebra $\sigma(X_T)$, and to set it as the single interpretive row (W3) of a dictionary whose other rows are theorems. We make no stronger claim for W3 than this: it is a measure-theoretic concretization of the inference-asymmetry of \cite{dibiagio2021}, congenial to the agent-relative (QBist) reading discussed in Sec.~\ref{sec:discussion}, and it yields no independent result---which is exactly why it sits below the rule in Table~\ref{tab:dictionary} and is offered as interpretation, not theorem. Existing instances where this asymmetry has been made quantitative include weak-value experiments \cite{ritchie1991,dressel2014}, past-quantum-state estimation in cavity and circuit QED \cite{gammelmark2013}, and quantum Schr\"odinger bridges connecting prescribed initial and final ensembles \cite{georgiou2025}.

\section{Relation to other frameworks}
\label{sec:discussion}

\emph{Two-state-vector formalism.} The TSVF \cite{aharonov1991} supplies the closed-system limit of Sec.~\ref{sec:quantum} and interprets the pair $(\langle\phi|, |\psi\rangle)$ ontologically. Our dictionary uses its statistics, Eq.~\eqref{eq:abl}, while remaining agnostic on the ontology; Corollary~\ref{cor:w2} is available to the TSVF as a consistency statement.

\emph{Consistent and decoherent histories.} The histories program of Griffiths, Omn\`es, and Gell-Mann and Hartle \cite{griffiths1984,omnes1992,gellmann1990} is the closest neighbor to the partition structure of W1. It assigns probabilities to sets of histories---partitions of the space of event sequences---subject to a consistency condition, is explicitly time-symmetric, and dispenses with both collapse and a fundamental role for measurement. Its single-framework rule, which forbids inferences that combine incompatible partitions, is the histories-level form of our Corollary~\ref{cor:w2}: the future choice of which partition to apply cannot rewrite an earlier marginal, because the ``which-path'' and ``interference'' descriptions are not jointly assertable. We work one level down, at the operational instrument--POVM description, and do not require the decoherence condition; W1 and the no-retrocausation boundary are, in that sense, the operational shadow of a single consistent set.

\emph{Transactional interpretation.} Cramer's advanced--retarded handshake \cite{cramer1986} shares the forward--backward geometry but posits a physical transaction. In the dictionary the backward object $\eff_t$ is inferential, on all fours with a backward information filter \cite{rauch1965}; no physical wave propagates backward.

\emph{QBism.} QBism \cite{fuchs2013} reads all quantum probabilities as an agent's conditional degrees of belief, which is congenial to the filtration reading of W3; the dictionary differs in isolating, via Proposition~\ref{prop:nosignal}, agent-independent marginal structure on which all agents must agree.

\emph{Stochastic--quantum correspondences.} Nelson's stochastic mechanics \cite{nelson1966,kuipers2023} and the Barandes correspondence \cite{barandes2023} seek to \emph{ground} quantum theory in stochastic processes. Our aim is orthogonal: we do not derive one theory from the other but exhibit a shared conditioning structure, Eq.~\eqref{eq:isomorphism}, that constrains temporal language in both.

\section{Conclusion}
\label{sec:conclusion}

Wheeler asked for \emph{it from bit}; the mathematics returns something more modest and more durable: \emph{conditioning without retrocausation}. The technical core is a single identity---marginalizing over the final measurement (the completeness of the POVM, $\sum_f F_f = \mathbbm{1}$, together with trace preservation of $\Lambda_2$) is the tower property of conditional expectation [Eq.~\eqref{eq:isomorphism}]---from which the whole dictionary follows. Read through it, participation is the choice of a partition (W1); delayed choice is post-selection, bounded by the marginal invariance of Proposition~\ref{prop:nosignal} (W2); the ``pull from the future'' is a bridge drift, real inside a conditioned description and absent from the unconditional one; and the arrow of time, on the interpretive reading we have marked as such, is the asymmetry of the filtrations an embedded observer inhabits (W3). The conditional calculus is time-symmetric; its causal structure is not.

That the operational content of delayed choice carries no retrocausation is, in the end, no revisionist reading of Wheeler. It is what Wheeler said when pressed (Sec.~\ref{sec:wheeler}), what Ellerman's separation-fallacy analysis concludes \cite{ellerman2015}, and what the single-framework rule of consistent histories enforces \cite{griffiths1984}; our contribution is to make it a one-line theorem and to exhibit its exact classical twin. The framework claims no more. It does not derive the Born rule, which it takes as an input on both sides, and it does not adjudicate the measurement problem, with which it is compatible either way. What remains is to sharpen the interpretive thesis of Sec.~\ref{sec:scope} into a quantitative measure of the filtration asymmetry, and to test the dictionary where it is most likely to strain---in the open-system, non-Markovian regime, where the quantum smoother and the classical bridge must be held to the same pairing, $\Tr[\eff_t\rho_t] \leftrightarrow \int\beta_t\,d\pi_t$. The dictionary of Table~\ref{tab:dictionary} is meant as a boundary stone: on one side a precise and usable participatory language, on the other the retrocausal overreach that Wheeler's poetry invites and his physics does not require.

\begin{acknowledgments}
\textit{Author contributions.}---C.-F.K.\ conceived the study, developed the theory and the symbolic and numerical verifications, and wrote the manuscript. K.-W.W.\ proposed the enlargement-of-filtrations framing of Wheeler's participatory universe. Both authors discussed the results and reviewed the manuscript. The authors received no specific funding for this work and declare no competing interests.
\end{acknowledgments}

\end{document}